\documentclass[leqno,12pt]{article} %leqno is the option to put formula numbers on the left side
\usepackage{amsmath, amssymb}
\usepackage{amsthm} 
\usepackage{amscd}
\usepackage{amsxtra}
\usepackage{xy}
\usepackage{xypic}
\define\cal{\mathcal}

\define\gl{\lambda}
\theoremstyle{plain} %text of this environment is typesetted in italics
\newtheorem{theorem}{\indent\sc Theorem}[section]
\newtheorem{lemma}[theorem]{\indent\sc Lemma}

\newtheorem{proposition}[theorem]{\indent\sc Proposition}

\theoremstyle{definition} %text of this environment is typesetted in roman letters
\newtheorem{definition}[theorem]{\indent\sc Definition}
\newtheorem{remark}[theorem]{\indent\sc Remark}
\newtheorem{example}[theorem]{\indent\sc Example}

\makeatletter
\def\address#1#2{\begingroup
\noindent\parbox[t]{7.8cm}{%
\small{\scshape\ignorespaces#1}\par\vskip1ex
\noindent\small{\itshape E-mail address}%
\/: #2\par\vskip4ex}\hfill%
\endgroup}%
\makeatother
\title{\uppercase{On eigenproblem for inverted harmonic oscillators}} %title of the paper
\author{
\bigskip \\
\textsc{ Piotr Kraso{\'n}, Jan Milewski } %names of authors
}
\date{} %leave empty
\begin{document}
\renewcommand{\theequation}{\thesection.\arabic{equation}}

\maketitle

%%%%%%%%%%%%%%% footnote %%%%%%%%%%%%%%%%
\footnote{ %2000 MSC numbers
2000 \textit{Mathematics Subject Classification}.
Primary 34L10; Secondary 33D15 ; 46F05 .
}
\footnote{ %key words and phrases
\textit{Key words and phrases}. 
inverted harmonic oscillator, rigged Hilbert space, generalized eigenvalue problem, differential operator.
}
%\footnote{ %acknowledgment of support etc. if any
%$^{*}$The research has been partially sponsored by the 
%research grant  of the Polish 
%National Center of Science (NCN)}

%%%%%%%%%%%%%%%%%%%%%%%%%%%%%%%%%%%%%%%%%

\begin{abstract}
\noindent
 We consider an eigenvalue problem for an inverted one dimensional harmonic oscillator.    We find a complete description  for the eigenproblem in $C^{\infty}(\mathbb R)$. The eigenfunctions are described in terms of the confluent hypergeometric functions, the spectrum is ${\mathbb C}$. The spectrum of the differential operator $-{\frac{d}{dx^2}}-{\omega}^{2}{x^2}$ is continuous and has physical significance  only  for the states which are in $L^{2}(\mathbb R)$ and correspond to real eigenvalues.  To identify them  we  orthonormalize in Dirac sense the states corresponding to real eigenvalues. This leads to the doubly degenerated real line as the spectrum of the Hamiltonian (in  $L^2({\mathbb R})$). We also use
    two other  approaches.  First we define a unitary operator between $L^{2}(\mathbb R)$ and $L^{2}$ for two copies of $\mathbb R$. This operator has the property that the spectrum of the image of the inverted harmonic oscillator corresponds to the spectrum of the operator $-i{\frac{d}{dx}}$. This shows again that the (generalized) spectrum of the inverted harmonic operator is a doubly degenerated real line. The second approach uses rigged Hilbert spaces. 
\end{abstract}

\section{Introduction} %delete * to number this section

In quantum mechanics one of the most important models is that of a harmonic oscillator. This is given by the following hamiltonian:
\begin{equation}\label{equation 1.0}
-\frac{d^2}{dx^2} + {{\omega}^{2}  x^{2}}
\end{equation} 
The quantum states are eigenfunctions of the eigenvalue problem:
\begin{equation}\label{equation 1.1}
-\frac{d^2}{dx^2}{\phi} + {{\omega}^{2}  x^{2}}{\phi}={\lambda}{\phi}
\end{equation}
The differential operator (\ref{equation 1.0}) is unbounded, symmetric and positive definite therefore the spectrum is discrete and the eigenvalues of (\ref{equation 1.1}) are non-negative real numbers which correspond to the quantized energy levels (cf. \cite{La}). Being unbounded the operator (\ref{equation 1.0}) is not defined on the whole $L^{(2)}(\mathbb R)$ but on its dense subset ${\cal S}(\mathbb R)$ of rapidly decreasing functions (cf. \cite{SR} Example 2 p. 250). 
By $L^2({\mathbb R})$ (resp. ${\cal S}(\mathbb R)$) we denote, for shorthand, $L^2({\mathbb R}, {\mathbb C})$ (resp. ${\cal S}(\mathbb R, {\mathbb C})$ ) - the space of square integrable (resp. rapidly decreasing) functions from ${\mathbb R}$ into ${\mathbb C}.$
Moreover the operator (\ref{equation 1.0})  as every symmetric operator is closable. One might also show that it is essentially self-adjoint (cf. \cite{SR}).

The inverted harmonic oscillator in stationary states is described by the following operator:
\begin{equation}\label{equation 1.2}
-\frac{d^2}{dx^2} - {{\omega}^{2}  x^{2}}
\end{equation}
The operator here is of course symmetric and therefore closable  but not positive definite.  We solve  the eigenproblem directly for the operator (\ref{equation 1.2}) in the space $C^{\infty}(\mathbb R).$ The spectrum in $C^{\infty}(\mathbb R)$ is continuous and equal to ${\mathbb C}$. The authors oppose, for physical reasons, the  treatment which involves complex (non-real) eigenvalues.

\begin{remark}
Recall that for the attractive oscillator (\ref{equation 1.0}) the eigenfunctions are given by means of Hermite polynomials
\begin{equation}\label{eigen1}
{\psi}_n(x:{\omega})=A_ne^{-{\omega}x^2/2}H_n(\sqrt{\omega}x), \qquad {\lambda}_n=(2n+1)\omega \quad \omega >0
\end{equation}
where $A_n(\omega)=\frac{1}{\sqrt{2^nn!}}\left(\frac{\omega}{\pi}\right)^{\frac{1}{4}}$ is a normalizing constant: $<{\psi}_n, {\psi}_m >={\delta}_{n,m}.$ Naturally, the linear combination ${\Sigma}_{n}a_n{\psi}_n(x; \omega)$ is a square integrable function  for 
${\Sigma}_{n}|a_n|^2<\infty .$
A change of $\omega$ into $i{\omega}$ leads to unbounded eigenvectors ${\psi}_{n}(x; i\omega)$, $||{\psi}_{n}(x; i\omega)||=\infty $ with the discrete purely imaginary spectrum ${\lambda}_n=(2n+1)i\omega$. 
Therefore any non-trivial  linear combination does not yield a square integrable function. These states cannot be viewed as 
quantum states. From physical point of view, the Hamiltonian of this system is the restriction of ({\ref{equation 1.2}}) to 
$L^2({\mathbb R}).$
\bigskip
\end{remark}

%In particular a simple change $\omega$ into $i\omega$  in the solution of non-inverted harmonic oscillator does not seem to be completely correct and in the spirit of quantum mechanics. 

In general,  any hermitian operator has real spectrum and the functions  with the non-real eigenvalues cannot be interpreted as describing quantum states. 
 One should look for eigenstates in $L^{2}(\mathbb R)$ and in a real part of a spectrum. The differential operator (\ref{equation 1.2}) is not defined on the whole $L^{2}({\mathbb R})$ but on its dense subspace  $D$. The spectrum depends of course on the choice of the domain. 
 We identify the real spectrum on $D$ by orthonormalizing the real states (cf. Proposition \ref{thm. 3.4}) in the Dirac sense.
 This uses fairly complicated integral transform.
 
 To identify the real spectrum of the closure of the operator (\ref{equation 1.2}) on $L^{2}(\mathbb R)$, in a different way, we define two unitary operators $W_{A}:L^2(\mathbb R)\rightarrow L^2(\mathbb R)$ 
(cf. ({\ref{opcal}}) and (\ref{jadrcalk})) and $U_{\exp}:L^{2}(\mathbb R) \rightarrow L^{2}({\mathbb R}_{+1}\amalg {\mathbb R}_{-1}),$ where ${\mathbb R}_{+1}\amalg {\mathbb R}_{-1}$ is the disjoint union of two copies of ${\mathbb R}.$ We prove that the operator (\ref{equation 1.2}) on $L^{2}({\mathbb R})$ corresponds to the operator $-i{\frac{d}{dx}}$ on $L^{2}({\mathbb R}_{+1}\amalg {\mathbb R}_{-1}).$ As the spectrum of the latter is real we see that the spectrum of (\ref{equation 1.2}) is real on the appropriate domain $D=W_{A}^{-1}U_{\exp}^{-1}(D_{1})$ where $D_{1}={\cal S}({\mathbb R}_{+1})\amalg {\cal S}({\mathbb R}_{-1})$ is the domain of the Fourier transform. In view of Theorem X.1 of \cite{SR1} this also gives an indirect proof that the closure of operator (\ref{equation 1.2}) is essentially self-adjoint.
We believe that our approach with unitary operators $W_A$
and $U_{\mathrm{exp}}$ is new, although in \cite{V} (cf. \cite{Mi} [) an alternative procedure is suggested. This
relies on the unitary transformation of the operator $x \frac{d}{dx}+\frac{1}{2}$ into the operator $\frac{d}{dx^2}+ \frac{x^2 }{4}$
Since the resolution of the first operator is known \cite{V} one obtains the spectrum for the latter.

The rigged Hilbert space approach \cite{Ge} is common for essentially self-adjoint operators and  was used in \cite{Ma1} where a  one-dimensional system with a  rectangular barrier potential was considered. The main challenge in applying the  theory of rigged Hilbert spaces is to find the appropriate  dense topological linear subspace ${\Phi}$ (cf. section \ref{section 2}). We define the rigged Hilbert space for the generalized eigenproblem for (\ref{equation 1.2}) and prove that the generalized spectrum is real. The naturality of our choice is justified by Lemmas \ref{ulemma} and \ref{Wlemma} (cf. Remark \ref{last}).

\section{ Eigenproblem for the differential operator}
In this section we analyze the eigenproblem for the operator (\ref{equation 1.2}).   We will describe  eigenvectors corresponding to an eigenvalue $\lambda$. These are solutions of the following equation:
\begin{equation}\label{eq. 3.1}
( -\frac{d^2}{d x^2}-{\omega}^2 x^2 ) \psi (x)= \lambda \psi  (x). 
\end{equation}
At first let us consider the eigenproblem (\ref{eq. 3.1}) in $C^{\infty}({\mathbb R}).$

Let 
\begin{equation}\label{eq. 3.2}
_{1}F_{1}(a,b,z)=1+{\frac{az}{b}}+{\frac{(a)_{2}z^{2}}{(b)_{2}2!}}+\dots + {\frac{(a)_{n}z^{n}}{(b)_{n}n!}}+\dots
\end{equation}
where
\begin{equation}\label{eq. 3.3}
(a)_{n}=a(a+1)(a+2)\dots (a+n-1)
\end{equation}
be a Kummer function (confluent hypergeometric function of type (1,1)\,) (cf. \cite{Ab}, \cite{Bu}, \cite{WW}.)
Define a Fresnel factor:
\begin{equation}\label{eq. 3.4}
 f_{\alpha} (x)= \exp ( i {\alpha} x ^2 /2)
\end{equation}
\begin{theorem}\label{thm. 3.5}
The spectrum of the operator (\ref{equation 1.2}) is continuous and equals $\mathbb C$. For an eigenvalue ${\lambda}\in {\mathbb C}$ the corresponding generalized linearly independent eigenvectors may be  given as:
\begin{equation}\label{eq. 3.6}
{\psi}_{P}(x)=f_{\alpha} (x) { _1 F  _1} \left( \nu +\frac{1}{4}; \frac{1}{2}; -i {\alpha} x^2 \right) 
\end{equation}
and
\begin{equation}\label{eq. 3.7}
{\psi}_{N}(x)=f_{\alpha} (x)x{ \, _1 F  _1} \left( \nu +\frac{3}{4}; \frac{3}{2} ; -i {\alpha} x^2 \right)
\end{equation}
where ${\alpha}=\pm {\omega}$ and $\nu =\frac{ \gl}{4 i \alpha}.$
\end{theorem}
\begin{proof}
Assume that 
\begin{equation}\label{eq. 3.8}
\psi (x)= f_{\alpha} (x) F (x)
\end{equation}
Substituting (\ref{eq. 3.8}) into (\ref{eq. 3.1}) and taking into account that
\begin{equation}\label{eq. 3.81}
f_{\alpha} '(x)= i {\alpha} x f_{\alpha} (x), \quad f_{\alpha} ''(x)=(i{\alpha} -{\alpha} ^2 x^2 ) f_{\alpha} (x)
\end{equation}
 one obtains  the following equation for $F(x):$
\begin{equation}\label{eq. 3.9}
 F\,''(x)=-2i{\alpha} x F\,'(x)+\left[\left({\alpha} ^2-{\omega} ^2 \right)x^2 -i{\alpha} -{\lambda} \right] F (x)
\end{equation}

Notice that for ${\alpha}=\pm {\omega}$ equation (\ref{eq. 3.9}) is of Hermite type. The general solution of (\ref{eq. 3.9}) for ${\alpha}=\pm {\omega}$ is given by the even solution:
\begin{equation}\label{eq. 3.10}
F_{P,{\alpha},\lambda}(x)= _1 F  _{1} \left( \nu +\frac{1}{4}; \frac{1}{2}; -i {\alpha} x^2 \right)
\end{equation}
and the odd solution:
\begin{equation}\label{eq. 3.11}
F_{N,{\alpha}, \lambda}(x)=x{ \, _1 F  _{1} \left( \nu +\frac{3}{4}; \frac{3}{2} ; -i {\alpha} x^2 \right)}
\end{equation}
and $\nu =\frac{ \gl}{4 i \alpha}.$
This can be verified by a direct computation using appropriate expansions (\ref{eq. 3.2}). 
\end{proof}
\begin{remark}\label{rem. 3.2}
Note that because of the uniqueness of the solution for  equation (\ref{eq. 3.1}) we have:
\begin{equation}\label{eq. 3.2a}
 f_{-{\omega}}(x)F_{P,{-\omega},{\lambda}}(x)=f_{{\omega}}(x)F_{P,{\omega},{\lambda}}(x)
\end{equation}
\begin{equation}\label{eq. 3.2b}
 f_{{-\omega}}(x)F_{N,{-\omega},{\lambda}}(x)= f_{{\omega}}(x)F_{N,{\omega},{\lambda}}(x)
\end{equation}
which are  also clear consequences of  Kummer's first formula (cf. \cite{Bu} p.6. formula 12).
\end{remark}

From now on we consider the Hamilton operator as the restriction of the differential operator to $L^2(\mathbb R).$

We have the following:
\begin{proposition}\label{thm. 3.4}
Let
\begin{equation}\label{eq. 3.13}
 F_a (x)=\frac{1}{\sqrt{x}}(A (a)e^{(ig_a(x))}+ \overline{A (a)}e^{(-i g_a(x))}) ,
\end{equation}
where 
 $$g_a (x)=\frac{{\alpha} x^2}{2}+\frac{a}{2{\alpha}}ln(x).$$
 Further, let ${Wr (\overline{F_b} , F_a)}$ be a Wronskian of $\overline{F_b}$ and $ F_a. $ Then for $a,b \in{\mathbb R}$ we have:
\begin{equation}\label{eq. 3.15}
\lim_{x\rightarrow \infty }\frac{Wr (\overline{F_b} , F_a)}{a-b}  =4\pi \alpha |A(a)|^2\delta (b-a).
\end{equation} 
where ${\delta}$ denotes the standard Dirac ${\delta}$-distribution on a real line.
\end{proposition}
Notice that the assumption that $a$ and $b$ are real is important because of the structure of equality (\ref{eq. 3.15}).
\begin{proof}
Let us denote
\begin{equation}\label{eq. 3.20}
 W_{s_1,s_2}:=Wr \left(\frac{1}{\sqrt{x}}\exp\left(s_1 i f_b(x)\right),
\frac{1}{\sqrt{x}}\exp\left(s_2if_a(x)\right)\right),
\end{equation}
where $s_{1,2}=\pm$.

Wronskian satisfies the following identity:
\begin{equation}\label{eq. 3.21}
 Wr (fg,fh)=f^2 Wr(g,h) .
\end{equation}
Hence
\begin{equation}\label{eq. 3.22}
  W_{-,+}=- \left(2i\alpha +i\frac{b+a}{2\alpha x^2}\right)\exp\left(i\frac{a-b}{2\alpha}\ln x\right)
\end{equation}
and
\begin{equation}\label{eq. 3.23}
 W_{+,+}= i\frac{b-a}{2\alpha x^2}\exp\left(i\alpha x^2+i\frac{a+b}{2\alpha }\ln x\right).
\end{equation}
Moreover
\begin{equation}\label{eq. 3.24}
 W_{-s_1,-s_2}=\overline{W_{s_1,s_2}} .
\end{equation}
Since the Wronskian is $2$-linear we have:
\begin{equation}\label{eq. 3.25}
 Wr (\overline{F_b} , F_a)=
 \overline{A (b)}A (a)W_{-,+}+\overline{A (b)}\overline{A (a)}W_{-,-}+ 
\end{equation}
 $$ A (b)A (a)W_{+,+}+A (b)\overline {A (a)}W_{+,-}.$$
Dividing (\ref{eq. 3.25}) by $(a-b)$ and passing to the limit as $x\rightarrow\infty$ we obtain (\ref{eq. 3.15}). Here we treat both sides of (\ref{eq. 3.25}) as distributions.  
\end{proof}
%\begin{proof}
The Kummer function has the following asymptotic expansion for $|z|>>0$ (cf. \cite{Ab} formulas 13.1.4 and 13.1.5   ):
\begin{equation}\label{eq. 3.16}
{ _1 F  _1} (a,b,z)=\frac{\Gamma (b)}{\Gamma (b-a)} (-z)^{-a}+\frac{\Gamma (b)}{\Gamma (a)}e^z (z)^{a-b}+O(\frac{1}{z})
\end{equation} 

Therefore for $x>>0$ we have:

\begin{equation}\label{eq. 3.17}
F_{(P,N),\alpha, \lambda } (x) \sim A_{P,N}(\alpha ,\lambda)\frac{1}{\sqrt{x}}\exp(\frac{i\alpha  x^2}{2}+\frac{i\lambda }{2\alpha}\ln x) +
\end{equation}
$$
\overline{ A_{P,N}(\alpha ,\lambda )}\frac{1}{\sqrt{x}}\exp(-\frac{i\alpha  x^2}{2}-\frac{i\lambda }{2\alpha}\ln x) 
$$
where 
\begin{equation}\label{eq. 3.18}
 A_P=\frac{\Gamma(\frac{1}{2})}{\Gamma(\frac{1}{4}+\frac{i\lambda }{4\alpha})}(i\alpha)^{-\frac{1}{4}+\frac{i\lambda}{4\alpha}}, 
\end{equation}
\begin{equation}\label{eq. 3.19}
 A_N=\frac{\Gamma(\frac{3}{2})}{\Gamma(\frac{3}{4}+\frac{i\lambda}{4\alpha})}(i\alpha)^{-\frac{3}{4}+\frac{i\lambda}{4\alpha}}, 
\end{equation}
and $F_{(P,N),{\alpha}, \lambda }$ is either $F_{P,{\alpha}, \lambda }$ or $F_{N,{\alpha}, \lambda}$ (cf. (\ref{eq. 3.10}) and (\ref{eq. 3.11})).

Let ${\phi}_1, {\phi}_2$ be the eigenfunctions of the differential operator $-{\frac{d}{dx^2}+U(x)}$ where $U(x)$ is a real valued function with real eigenvalues ${\lambda}_1$ and ${\lambda}_2.$ 
The following formula allows one to orthonormalize states  for continuous spectra in the Dirac sense. (Analogously for discrete spectrum in the Kronecker sense)
\begin{equation}\label{wron1}
<{\phi}_1, {\phi}_2> = {\int}_{-\infty}^{\infty}{\bar{{\phi}_1}}{{\phi}_2}dx={\frac{1}{{\lambda}_1 - {\lambda}_2}}Wr(\bar{ {\phi}_1}, {\phi}_2)|_{-\infty}^{+\infty}
\end{equation}
 Notice that in our case Proposition \ref{thm. 3.4} the asymptotic form of the Kummer function (\ref{eq. 3.17}) and the equality (\ref{wron1}) allows one to orthonormalize the states in the Dirac sense. 

Since the operator (\ref{equation 1.2}) is hermitian in $L^2({\mathbb R})$ we have to take the real part (subset of real eigenvalues) of the above, computed in $C^{\infty}({\mathbb R}),$ spectrum. By Proposition \ref{thm. 3.4} and (\ref{wron1}) we can orthonormalize in $L^2(\mathbb R)$  the states corresponding  to real eigenvalues and thus  we get doubly degenerated spectrum ${\mathbb R}.$
In order to illustrate this we define, in the next section, a unitary operator from $L^2({\mathbb R})$ into $L^2({\mathbb R}).$ The transformed operator (\ref{equation 1.2}) has the same real part of the spectrum as  (\ref{equation 1.2}). Notice that assumption in Proposition  (\ref{thm. 3.4}) that $a$ and $b$ are real is in agreement with the fact that we consider the operator (\ref{equation 1.2} ) on $L^2({\mathbb R})$ where it is hermitian.
\begin{remark}
With any real eigenvalue $\lambda$ of the Hamiltonian (\ref{equation 1.2}) there are associated two eigenstates: one is an even function and the second is an odd function (cf. (\ref{eq. 3.10}) and (\ref{eq. 3.11}) ).
\end{remark}

\section{Transformation of position and momentum operators}

 Position and momentum operators act on a function $f(x)$ in the Schr{\"o}dinger representation in the following way \cite{La}:
\begin{equation}\label{eq. 3a1}
 \hat x f (x)=x f(x), \quad \hat p_x f (x)=-i\frac{d}{dx} f (x) .
\end{equation}
One readily verifies that these operators fulfil the following commutation relation:
\begin{equation}\label{eq. 3a2}
[\hat p_x,  \hat x]=-i \,{\rm Id}
\end{equation} 
For a $2\times 2$-matrix:
\begin{equation}\label{eq. 3a3}
 A= \left[ \begin{array}{cc} a & b \\ \alpha & \beta \end{array} \right] 
\end{equation}
we define the following transformation of operators:
\begin{equation}\label{eq. 3a4} 
\hat u=a\hat x  +b\hat p_x, \quad \hat p_u=\alpha \hat x  +\beta \hat p_x .
\end{equation}
The operators (\ref{eq. 3a4}) fulfil the relation (\ref{eq. 3a2}) if and only if $A \in SL_{2}(\mathbb R).$ In this case one can view $\hat u$ and $\hat p_u$ as new operators of position and momentum. 
Now we will look for the unitary transformation $W_A$ which maps the Schr{\"o}dinger representation for  operators $\hat x$ and $ \hat p_x$ into that for new operators $\hat u$ and $ \hat p_u .$ In other words we would like $W_A$ to  fulfil the following equations:

\begin{equation} \label{transfxnau}
W_A(\hat u f) (u)=uW_A( f) (u), \quad W_A(\hat p_u f) (u)=-i \frac{d}{du} W_A( f) (u).
\end{equation}
Note that the transformation $W_{A}$ is not uniquely defined but only up to a complex constant $c$ of absolute value one. 
The following proposition  shows that the transformation $W_A$ is given by  an integral operator.
\begin{proposition}\label{prop}
Any unitary transformation $W_{A}$ is given by the formula:
\begin{equation}\label{opcal}
 ({W_A} f) (u)=\int_{\mathbb R} W (A; u,x) f(x) dx ,
\end{equation}
where  the integral kernel is of  the form
\begin{equation}\label{jadrcalk} 
W (A; u,x)= c(A)\exp \left[\frac{i}{2b} (ax^2-2ux+{\beta} u^2)\right] 
\end{equation}
and
\begin{equation}\label{coda}
 |c(A)|=\sqrt{\frac{1}{2{\pi}  b}}.
 \end{equation}
\end{proposition}
\begin{proof}
Let ${\hat K}\in <{\hat x},{\hat p}>$ be an element of the algebra generated by ${\hat x}$ and ${\hat p}.$  Let ${\hat K}^{T}$ be an operator defined by the formula 
\begin{equation}
{\int}_{\mathbb R}({\hat K}^{T}f(x))g(x)dx={\int}_{\mathbb R}f(x)({\hat K}g(x))dx
\end{equation}
Integrating by parts one obtains:
\begin{equation}\label{equation{3a5}}
 W_A(\hat K f) (u)=\int_R ({\hat K}^T W (A; u,x)) f(x) dx, 
\end{equation}
where  ${\hat K}^T$  acts on  $W (A; u, x)$ as a function of $x$. 
In particular we have
\begin{equation}
({ {\alpha}\frac{d}{dx}})^{T}=-{\alpha}\frac{d}{dx}  \quad\text{and} \quad {\hat x}^{T}={\hat x} 
\end{equation}
and ${\hat u}^{T}=(a{\hat x}+ b{\hat p}_{x})^{T} = a{\hat x}- b{\hat p}_{x.}$

 We also have
\begin{equation}
 {\hat u}^T W(A;u,x)=u W(A;u,x) 
\end{equation}
and 
\begin{equation}
{\hat p_u}^T W(A;u,x)=\left(-\frac{1}{b}x+\frac{\beta}{b}\right) W(A;u,x)= 
-i \frac{\partial}{\partial u}W(A,u,x) 
\end{equation}
which gives (\ref{transfxnau}) and (\ref{jadrcalk}).
Notice that 
\begin{equation}\label{rozklad}
W_{A}=M_{2}{\cal F}M_{1}
\end{equation}
where 
\begin{equation}\label{rozkl. 1}
(M_{1}f)(x)=e^{{\frac{i}{2b}ax^2}}f(x) 
\end{equation}
\begin{equation}\label{rozkl 2}
 ({\cal F}f)(u)=c(A){\int}_{\mathbb R}e^{\frac{-iux}{b}}f(x)dx 
\end{equation}
\begin{equation}\label{rozkl 3}
 (M_{2}f)(u)=e^{{\frac{i}{2b}}{\beta}u^2}f(u)
\end{equation}
 The unitarity of the operator (\ref{rozkl 2}) gives (\ref{coda})
. \end{proof}
We have shown so far that there exists an integral unitary transformation of the form
(\ref{transfxnau}) with the integral kernel (\ref{jadrcalk}). The following lemma guarantees the uniqueness of such a transformation. 

\begin{lemma}\label{homomorphism}
Assume that for $A_{1},A_{2} \in SL_{2}({\mathbb R})$ the corresponding unitary operators fulfil the following
\begin{equation}\label{eq. act}
 W_{A_{2}} W_{A_{1}}=W_{A_{2}A_{1}}. 
\end{equation}
Then 
\begin{equation}\label{eq2222}
 c(A)=\sqrt{\frac{1}{2{\pi}i b}}.
\end{equation}
and therefore $W_{A}$ is defined uniquely.
\end{lemma}

\begin{proof}
Notice that  $c(A)=\sqrt{\frac{1}{2{\pi}i b}}$ yields $W(A_{2}A_{1} y,x)={\int}_{\mathbb R}W(A_{2}; y, u)W(A_{1}; u, x)dx$ and therefore (\ref{eq. act})
  Any other $W^{\prime}_{A}$ with the constant $C^{\prime}(A)$ fulfilling (\ref{eq. act}) yields $C^{\prime}(A)={\Gamma}(A)C(A)$ where $\Gamma$ is a one dimensional representation of $SL_2({\mathbb R})$. It is well known \cite{Lan} that the only such representation is the trivial one.
\end{proof}
We call the  functions $f$ which are the arguments of $W_{A}$  originals and $W_{A}f$ the  images. 
Consequently, we have the transformation rule for operators ${\hat K}\rightarrow {\hat K}_{A}=W_{A}{\hat K}W_{{A}^{-1}}.$  In particular the Schr{\"o}dinger representation (\ref{transfxnau}) may be written in the form:
\begin{equation}
{\hat u}_{A}f_{A}(u)=uf_{A}(u)      \qquad ({\hat p}_{u})_{A}f_{A}(u)=-if^{\prime}_{A}(u)
\end{equation}

\begin{lemma}\label{calkGauss}
 The following equality holds:
\begin{equation}\label{equation{3.6}}
 \int_{0}^{\infty} \exp(-at^2 +tu) t^{\alpha} dt = a^{\frac{-1+{\alpha}}{2}}\Psi _{\alpha}(\frac{u}{\sqrt a}),
 \end{equation}
where
\begin{equation}\label{equation 3.7}
 {\Psi}_{\alpha} (v)=
\frac{1}{2}\Gamma\left(\frac{\alpha +1}{2}\right) { _1 F  _1}\left(\frac{\alpha +1}{2}, \frac{1}{2},\frac{v^2}{4}\right)+ 
 \frac{1}{4}\Gamma\left(\frac{\alpha+2}{2}\right)
v{\,  _1 F  _1}\left(\frac{\alpha +2}{2}, \frac{3}{2},\frac{v^2}{4}\right)
\end{equation}
\end{lemma}
\begin{proof}
Expanding into a power series  and integrating term by term we obtain (\ref{equation 3.7})
\end{proof}
\begin{lemma}\label{Wlemma}
We have the following equality:
%\begin{equation}\label{equation{3.a.71}}
% W_{A} x^{-\frac{1}{2}+i\gamma}=C(A)\exp\left( \frac{i\beta}{2b} u^2 \right) \l%eft( \frac{2i b}{a} \right)^{\frac{1}{4}+i\frac{\gamma}{2}}  \Phi _{-\frac{1}{2%}+i\gamma}\left(\sqrt{ \frac{2i b}{a}} u \right)
%\end{equation}
\begin{equation}\label{equation{3.a.7}}
 W_{A^{-1}} u^{-\frac{1}{2}+i\gamma}=C(A^{-1})\exp\left(- \frac{ia}{2b} x^2 \right) \left( \frac{-2i b}{\beta} \right)^{\frac{1}{4}+i\frac{\gamma}{2}}  \Psi _{-\frac{1}{2}+i\gamma}\left(\sqrt{ \frac{-2i b}{\beta}} x \right)
\end{equation}

where $u^{-\frac{1}{2}+i\gamma}$ is an eigenfunction of $({\hat u}{\hat p_{u}}+{\hat p_{u}}{\hat u})_{A}$ corresponding to the eigenvalue $\lambda = 2\gamma$.
\end{lemma}
\begin{proof}
Follows from Lemma \ref{calkGauss}
\end{proof}

Equation (\ref{equation{3.a.7}}) describes an eigenfunction of the operator  
\begin{equation}\label{eq. 3.7a1}
{\hat u}{\hat p_u}+{\hat p_u}{\hat u}= E{\hat x}^2 +F({\hat x}{\hat p}_x+{\hat p}_x {\hat x}) +G{\hat p}_x ^2
\end{equation} for 
\begin{equation}\label{quadr}
E= 2a{\alpha},\qquad F=a{\beta}+b{\alpha},\qquad G=2b{\beta}.
\end{equation}
\begin{remark}\label{spr}
 An elementary calculation shows that for any matrix
\begin{equation}
Q=\left [\begin{array}{ll}
E & F\\
F & G
\end{array}   
\right ].
\end{equation}
with ${\det Q}=-1$ there exists a matrix $A$ for  which (\ref{eq. 3.7a1}) holds.
In fact the choice of $A$ depends on only one parameter. For $E\neq 0$ and a parameter $a\neq 0$ we can choose:
\begin{equation}
A=\left [\begin{array}{ll}
a & {\frac{a(F-1)}{E}}\\
\frac{E}{2a} & \frac{F+1}{2a}
\end{array}   
\right ].
\end{equation} 
\end{remark}
For the matrix
\begin{equation}\label{eq.osc}
A=\left [\begin{array}{ll}
a & {\frac{a}{\omega}}\\
\frac{-{\omega}}{2a} & \frac{1}{2{\omega}}
\end{array}   
\right ].
\end{equation}  
we obtain the Hamiltonian $\hat H$ of the inverted harmonic oscillator:
\begin{equation}\label{eq.ham}
{\omega}({\hat u}{\hat p}_{u}+{\hat p}_{u}{\hat u}) = {\hat p}_{x}^{2}-{\omega}^{2}{\hat x}^{2}
\end{equation}

\begin{definition}
Let $U_{\text{exp}}^{\pm}: L^{2}({\mathbb R}_{\pm},dx) \rightarrow L^{2}({\mathbb R},dt)$ be  transformations  $f\rightarrow {\tilde f}_{\pm}$, where ${\tilde f}_{\pm}(t)=f(\pm e^{t})\sqrt{e^{t}}.$
\end{definition}

\begin{proposition}
 We have:
\begin{equation}
 U_{{\exp}}^{\pm}  ({\hat x}{\hat p}_{x}+{\hat p}_{x}{\hat x})=2{\hat p}_{t}
\end{equation}
\begin{equation}
U_{\exp}^{\pm} |x|^{-\frac{1}{2}+i\gamma}={\exp}(i{\gamma}t)
\end{equation}
\end{proposition}
\begin{proof}
Straightforward calculation.
\end{proof}
Note that  both operators $U_{\exp}^{\pm}$ are  unitary.
Define:
\begin{equation}\label{eq. heavy}
f_{\pm 1}(x)={\eta}(\pm x)f(x)
\end{equation}
where $\eta$ is the Heavyside function.
We have the decomposition:
\begin{equation}
f(x)=f_{+1}(x)+f_{-1}(x)
\end{equation}
Let ${\mathbb R}_{+1}= {\mathbb R}\times {\{+1\}}$, ${\mathbb R}_{-1}= {\mathbb R}\times \{-1\}$ and ${\mathbb R}_{+1}\amalg {\mathbb R}_{-1}$ be the disjoint union of two copies of ${\mathbb R}.$  The operators $U^{\pm}_{\exp}$ define an operator:
\begin{equation}\label{eq.dec}
U_{\exp}: L^{2}(\mathbb R_{+})\oplus L^{2}(\mathbb R_{-}) \cong L^{2}(\mathbb R) \rightarrow L^{2}({\mathbb R}_{+1}\amalg {\mathbb R}_{-1})
\end{equation}
Let 
\begin{equation}
{\phi}_{a,b,{\gamma}}(x)=a{\eta}(x)x^{-\frac{1}{2}+i{\gamma}} +b{\eta}(-x)x^{-\frac{1}{2}+i{\gamma}}
\end{equation}
\begin{lemma}\label{ulemma}
The following equality holds:
\begin{equation}\label{eq. ulemma}
U_{\exp}({\phi}_{a,b,{\gamma}})= {\cal F}_{a,b,{\gamma}}
\end{equation}
where
\begin{equation}\label{eq. ulemma1}
{\cal F}_{a,b,{\gamma}}(t, +1) = ae^{i{\gamma}t}  \quad \text{and} \quad {\cal F}_{a,b,{\gamma}}(t, -1) = be^{i{\gamma}t}
\end{equation}
\end{lemma}
\begin{proof}
Straightforward calculation
\end{proof}
\begin{definition}\label{eq. 2.10}
Let $A$ be an operator in a linear topological space ${\Phi}.$ A linear functional $F\in {\Phi}^{\prime}$ is called a generalized eigenvector of $A$ corresponding to the eigenvalue $\lambda$ if 
\begin{equation}\label{eq. 2.11}
F(A{\phi})={\lambda}F(\phi)
\end{equation}
for all ${\phi}\in {\Phi}.$
\end{definition}
\begin{theorem}\label{main theorem}
The eigenproblem of the operator( \ref{equation 1.2}) viewed as  the generalized eigenproblem on  the Hilbert space $L^{2}(\mathbb R)$ has a spectrum equal to ${\mathbb R}$ and the generalized eigenfunctions are given by linear combinations of (\ref{eq. 3.10}) and (\ref{eq. 3.11}).
\end{theorem}
\begin{proof}
We have defined unitary transformations $W_{A}$ and $U_{\exp}$ whose composition transforms the Hamiltonian of the inverted harmonic oscillator ${\hat p}^{2}-{\omega}^{2}{\hat x}^{2}$ into $2{\omega}p_{t}$. The generalized eigenfunctions for $p_{t}$ are given by $e^{i{\gamma}t}$, and the corresponding linear functional is the Fourier transform and therefore the generalized spectrum is the  real axis. Since the unitary transforms are isometries the theorem follows.
\end{proof}
\begin{remark}\label{remarkdom}
In Theorem \ref{main theorem}   the domain of the operator (\ref{equation 1.2}) is given by the following formula:
\begin{equation}\label{eq. dom}  
D= W_{A}^{-1}U_{\exp}^{-1}({\cal S}({\mathbb R}_{+1})\amalg {\cal S}({\mathbb R}_{-1})).
\end{equation}
\end{remark}

\section{Rigged Hilbert spaces}\label{section 2}
We start this section with some necessary definitions. We follow the approach of \cite{Ge}. We assume  our vector spaces to be  defined  over either  ${\mathbb R}$ or  ${\mathbb C}.$
\begin{definition}\label{definition 2.1}
A linear topological space $\Phi$ is called countably normed if there exist norms $||{\phi}||_{m}$  on $\Phi$ for a natural $m$ such that they are compatible i.e. if a sequence $\{\phi\}_{k}$ tends to zero in  the norm $||{\phi}||_{m}$ and the sequence is fundamental in the norm $||{\phi}||_{n}$ then  converges to zero in the norm $||{\phi}||_{n}$. The topology on $\Phi$ is then defined by the basis of neighbourhoods of zero $U_{n,\epsilon}=\{{\phi}: ||{\phi}||_{n}<{\epsilon}\}.$
\end{definition}

\begin{definition}\label{def. 2.2}
 A space $\Phi$  in definition \ref{definition 2.1} is called countably Hilbert if norms $||{\phi}||_{n}$  come from scalar products $(\phi , \phi )_{n}$ and $\Phi$ is complete.
\end{definition}
\begin{remark}\label{remark 2.3}
One can always assume that the scalar products fulfil the following inequalities:
\begin{equation}\label{eq. 2.31}
(\phi , \phi )_{1}\leq (\phi , \phi )_{2}\leq \dots  
\end{equation}
for any ${\phi}\in{\Phi}$.
\end{remark}

\begin{remark}\label{remark 2.4}
Let $\Phi$ be a countably Hilbert space. If we define  ${\Phi}_{n}$ as the completion of $\Phi$ with respect to the norm 
$||\phi||_{n}=\sqrt{(\phi , \phi )_{n}}$ then ${\Phi}={\bigcap}_{n=1}^{\infty}{\Phi}_{n}$. For the adjoint space we have: ${\Phi}^{\prime}={\bigcup}_{n=1}^{\infty}{\Phi}_{n}^{\prime}$.
\end{remark}

\begin{definition}\label{def. 2.5}
Let $\Phi$ be a countably Hilbert space. Let $n\ge m$ and $I_{m}^{n}: {\Phi}_{n}\rightarrow {\Phi}_{m}$ be a natural (continuous in view of remark \ref{remark 2.3}) inclusions. ${\Phi}$ is called nuclear if the maps $I_{m}^{n}$ have the form:\begin{equation}\label{eq. 2.51} 
 I_{m}^{n}(\phi)={\sum}_{i=1}^{\infty}  {\lambda}_{i}(\phi , {\phi}_{i})_{n}{\psi}_{i}\end{equation}   
 where ${\phi}_{i}$ (resp. ${\psi}_{i}$) are orthonormal systems of vectors in ${\Phi}_{n}$ (resp. ${\Phi}_{m}$ ), ${\lambda}_{i}>0$ and ${\sum}_{i=1}^{\infty}{\lambda}_{i} <{\infty}$. 
\end{definition}
Assume that in countably Hilbert nuclear space $\Phi$ there is defined another scalar product $\widehat{(\phi ,\psi)}$, continuous in each variable. Let $H$ be the completion of $\Phi$ with respect to  this scalar product. Let ${\Phi}^{\prime}$ be the adjoint space to $\Phi$. Obviously we have an inclusion of adjoints $H^{\prime}\subset {\Phi}^{\prime}.$ Since every linear  functional on $H$ is given
by the formula $h^{\prime}(h)={\widehat{(h,h_{1})}}$ we can identify $H$ with $H^{\prime}.$ Notice that the inclusion ${\Phi}\rightarrow H$ is linear whereas the adjoint to the inclusion map $ H\rightarrow {\Phi}^{\prime}$ is antilinear (resp. linear) if our spaces are defined over ${\mathbb C}$ (resp. ${\mathbb R}$).

\begin{definition}\label{def. 2.6}
A triplet of spaces ${\Phi} \subset H \subset {\Phi}^{\prime}$, where $\Phi$ is a countably Hilbert nuclear space, $H$ a completion of $\Phi$ with respect to a scalar product $\widehat{(\phi, \psi)}$ and ${\Phi}^{\prime}$ the adjoint of ${\Phi}$ is called a rigged Hilbert space.
\end{definition}

The following example \cite{Ge} describes a typical  situation where rigged Hilbert spaces appear
\begin{example}\label{ex. 2.7}
Let $L$ be a symmetric positive definite differential operator acting on a space $K$ of infinitely differentiable functions with bounded supports in a domain $\Omega$.
Define  scalar products in K by the formulas:
\begin{equation}\label{eq. 2.71}
 (\phi , \psi )_{n}={\sum}_{i=1}^{n}{\int L^{i}{\phi}(x){\overline{{\psi}(x)}}}dx    \,\,\,\qquad n=0,1,\dots 
\end{equation}
Let ${\Phi}_{n}$ be a completion of $K$ with respect to $ (\phi , \psi )_{n}$. In this way  putting $H={\Phi}_{0}=L^{2}(\Omega)$, ${\Phi}={\bigcap}_{n=1}^{\infty}{\Phi}_{n}$ and ${\Phi}^{\prime}={\bigcup}_{n=1}^{\infty}{\Phi}_{n}^{\prime}$ we obtain a rigged Hilbert space. $\Phi$ is of course the space of Schwartz functions i.e. functions that rapidly decrease with all their derivatives \cite{Ho}.
\end{example}
\begin{example}\label{ex. 2.71}
Let ${\Omega}={\mathbb R}$ and ${\Phi}\subset H \subset {\Phi}^{\prime}$ be as in example \ref{ex. 2.7}.
The Dirac distribution ${\delta}(x-h)$ can be viewed as an element of ${\Phi}^{\prime}$ since ${\int}_{\mathbb R}{\phi}(x){\delta}(x-h)dx= {\phi}(h)$ is a linear functional on ${\Phi}.$ The function $e^{i{\lambda}x}$ can be viewed as an element of ${\Phi}^{\prime}$ since ${\int}_{\mathbb R} {\phi}(x)\overline{e^{i{\lambda}x}}dx=F({\lambda})$, where $F$ is the Fourier transform of ${\phi}$.
\end{example}
\begin{remark}\label{rem. 2.8}
In example \ref{ex. 2.7} we have $(\phi , \phi )_{1}\leq (\phi , \phi )_{2}\leq \dots $ since $L$ is a positive definite symmetric operator. An easy proof shows also that $L$ transforms $\Phi$ into $\Phi$.
\end{remark}
\begin{remark}\label{rem. 2.9}
 Note that in example \ref{ex. 2.7} we could use integration with respect to any positive measure $d{\mu{(x)}}$. In this case we of course have ${\Phi}_{0}=H=L^{2}_{\mu}({\Omega}).$ 
\end{remark}

Now we will expand a bit  on Definition \ref{eq. 2.10}.

\noindent
 Notice that formula (\ref{eq. 2.11}) can be written as
\begin{equation}\label{eq. 2.11a}
A^{\prime}F={\lambda}F
\end{equation}
If $\lambda$ is an eigenvalue of $A,$ then denote  ${\Phi}_{\lambda}^{\prime}$ as the eigenspace of $A$ corresponding  ${\lambda}.$

Let ${\phi}\in {\Phi},$ $\lambda$ be a number and $F_{\lambda}\in {\Phi}_{\lambda}^{\prime}.$ One can define a linear functional $\widetilde{{\phi}_{\lambda}}$ by the formula:
\begin{equation}\label{eq. 2.9}
\widetilde{{\phi}_{\lambda}}(F_{\lambda})=F_{\lambda}(\phi)
\end{equation}
\begin{definition}\label{def. 2.12}
The correspondence ${\phi}\rightarrow {\widetilde{{\phi}_{\lambda}}}$ is called the spectral decomposition of ${\phi}$ with respect to the operator $A.$ The set of generalized eigenvectors of the operator $A$ is complete if $\widetilde{{\phi}_{\lambda}}\cong 0$ implies ${\phi}=0.$ 
\end{definition}
 Theorem $5^{\prime}$  , Sec. 4.5, Chap. I of \cite{Ge} asserts that a self-adjoint operator in a rigged Hilbert space ${\Phi}\subset H \subset {\Phi}^{\prime}$ has a complete system of generalized eigenvectors, corresponding to real eigenvalues.

\medskip

For the operator (\ref{equation 1.2}) we define ${\tilde{\Phi}}_n = {\cal E}\cap L^{2}_{{\mu}_{n}}({\mathbb R})\cap L^{2}({\mathbb R}),$ where $\cal E$ is the space of infinitely differentiable complex functions on ${\mathbb R}$ with the topology of compact convergence in all derivatives (cf. \cite{Sch}  Ch. III, Sec. 8, Example 3) and
\begin{equation}\label{equation 2.13} 
$$ \[ {\mu}_{n}=\left\{ \begin{array}{ll}
 1 & \mbox{if $|x|\le 1$} \\
1+ |x|^{-{\frac{1}{n}}} & \mbox{if $|x|>1$}
\end{array} \right. \]$$
\end{equation} 

Notice that  for  $n>m$ one has $1<{\mu}_{m} <{\mu}_{n}$  and therefore
\begin{equation}\label{2.13a}
  \dots \subset L^{2}_{{\mu}_{k}}(\mathbb R)   \subset      \dots  \subset  L^{2}_{{\mu}_{1}}(\mathbb R)   \subset   L^{2}(\mathbb R)                                  \end{equation}    
Let ${\Phi}_{0}= L^{2}({\mathbb R})$ and let ${\Phi}_{n}$ be the completion of ${\tilde{\Phi}}_{n}$ with respect to the following scalar product:
\begin{equation}\label{equation 2.14}
(\phi , \phi )_{n}= {\sum}_{k=0}^{n} <\phi , \phi >_{k} ,
\end{equation}
where $<\cdot , \cdot >_{k}$  for $k\geq 1$ is the scalar product in  $L^{2}_{{\mu}_{k}}({\mathbb R})$  and for $k=0$ it is the scalar product in $L^{2}({\mathbb R}).$
In this way we have constructed a countably Hilbert space:
\begin{equation}\label{equation 2.14 a}
{\Phi}={\bigcap}_{n=1}^{\infty}{\Phi}_{n}
\end{equation}
Since the space ${\cal E}$ is nuclear (cf. \cite{Sch},  Ch. III, Sec. 8, Example 3) and the subspaces of nuclear spaces are nuclear (cf. \cite{Sch}, Ch. III, Theorem 7.4)  we have the rigged Hilbert space:
\begin{equation}\label{equation 2.15}
{\Phi} \subset {\Phi}_{0} = L^{2}(\mathbb R) \subset {\Phi}^{\prime} 
\end{equation}
Now we can prove:
\begin{theorem}\label{theorem 2.15}
The operator $- {\frac{d^2}{dx^{2}}}-{\omega}^{2}x^{2} $ on a rigged Hilbert space (\ref{equation 2.15}) has real  eigenvalues. For any ${\lambda}\in {\mathbb R}$ the generalized eigenvectors are linear functionals corresponding to the functions given by the  formulas (\ref{eq. 3.6}) and (\ref{eq. 3.7}).
\end{theorem}
\begin{proof}
Examining the formula ({\ref{eq. 2.11a}}) for $A=- {\frac{d^2}{dx^{2}}}-{\omega}^{2}x^{2}$  the argument similar to that used in the proof of  Theorem  \ref{thm. 3.5} shows that if $F\in {\Phi}^{\prime}$ is a generalized eigenvector of an operator $A,$ it corresponds to linear functionals associated with the  functions given by the formulas (\ref{eq. 3.6}) and (\ref{eq. 3.7}).
Notice that $\Phi$ consists of functions $f$ such that for every ${\epsilon}>0,$ 
$f\in L^{2}_{{\mu}_{\epsilon}}(\mathbb R)$ where 
\begin{equation}\label{equation 2.16}
$$ \[  {\mu}_{\epsilon}= \left\{ \begin{array}{ll}
1 & \mbox{if $|x|\le 1$}\\
1+|x|^{-{\epsilon}} & \mbox{otherwise}
\end{array}
\right.  \] $$
\end{equation}
Examining the asymptotic behaviour of the  eigenvectors $F_{P,N,{\alpha}, \lambda}$ given by (\ref{eq. 3.17}) we see that for a real eigenvalue $\lambda$ the integral
\begin{equation}\label{equation 2.17}
\int_{\mathbb R}F_{P,N,{\alpha}, \lambda}\bar{\phi}dx
\end{equation}
is convergent for any ${\phi}\in \Phi$. This is a consequence of the convergence for any values of $b\ge 1$ and $\epsilon >0$ of the following improper integral: 
\begin{equation}\label{equation 2.18}
\int_{b}^{\infty}\frac{x^{\frac{1}{n}}}{1+x^{\frac{1}{n}}}x^{-1-{\epsilon}}dx
\end{equation}

 If ${\text{Im}}(\lambda)\neq 0$ there exists ${\phi} = {\phi}_{\lambda}$ such that (\ref{equation 2.17}) is divergent.
In fact it is enough to take ${\phi}_{\lambda}\in {\Phi}_{n}$ for $\frac{1}{n}< |{\text{Im}}(\lambda)|$ to see that the integral:
\begin{equation}\label{equation{2.18}}
{\int_{b}^{\infty}\frac{x^{\frac{1}{n}}}{1+{x^{\frac{1}{n}}}}x^{-1-\frac{1}{n} + |{\text{Im}(\lambda}|}}dx
\end{equation}
 is divergent and therefore $\lambda$ with  ${\text{Im}}(\lambda)\neq 0$ is not a generalized eigenvalue for a rigged Hilbert space (\ref{equation 2.15}).  
\end{proof}
\begin{remark}\label{last}
Lemma \ref{ulemma} and Lemma \ref{Wlemma} show that $e^{i{\gamma}t}$ for ${\gamma}\in {\mathbb R}$ are transformed into the functions $f(x)$  of order $O(|x|^{-\frac{1}{2}}).$ These yield well-defined elements of ${\Phi}^{\prime}$ of (\ref{equation 2.15}).
\end{remark}

%%%%%%%%%%%% Authors' addresses %%%%%%%%%%%%%

\address{% First  Author
Department of Mathematics \\
Szczecin University \\
70-415 Szczecin \\
Poland
}
{krason@wmf.univ.szczecin.pl}

%start a new line
\address{% Second Author
Institute of Mathematics\\
Pozna\'n University of Technology\\ 
 60-965 Pozna\'n\\
Poland 
}
{jsmilew@wp.pl}

\end{document}